\documentclass[12pt,a4paper]{amsart}

\usepackage{latexsym}
\usepackage{amsmath}
\usepackage{amssymb}
\usepackage{amsthm}
\usepackage{amsxtra}
\usepackage{amsfonts}
\usepackage{mathrsfs}
\usepackage{amsbsy}
\usepackage{graphicx}
\usepackage{hyperref}
\usepackage{mathtools}
\usepackage{stackengine}

\usepackage{color}

\newtheorem{teo}{Theorem}[section]
\newtheorem{lem}[teo]{Lemma}

\newtheorem{prop}[teo]{Proposition}
\newtheorem{remark}[teo]{Remark}

\numberwithin{equation}{section}

\renewcommand{\Im}{\operatorname{Im}\,}
\renewcommand{\Re}{\operatorname{Re}\,}

\newcommand{\HH}{\mathcal H}

\newcommand{\GG}{\mathscr G}

\newcommand{\B}{\mathcal B}
\newcommand{\F}{\mathcal F}

\newcommand{\n}{\noindent}

\newcommand{\ve}{\varepsilon}

\newcommand{\al}{\alpha}
\newcommand{\ga}{\gamma}

\newcommand{\la}{\lambda}

\newcommand{\ome}{\boldsymbol{\omega}}
\newcommand{\f}{\frac}
\newcommand{\ba}{\begin{eqnarray}} \newcommand{\ea}{\end{eqnarray}}
\newcommand{\be}{\begin{equation}} \newcommand{\ee}{\end{equation}}
\newcommand{\bdm}{\begin{displaymath}} \newcommand{\edm}{\end{displaymath}} 
\newcommand{\brr}{\begin{array}}\newcommand{\err}{\end{array}}

\newcommand{\lf}{\left}
\newcommand{\ri}{\right}

\newcommand{\bml}{\begin{gather}} 
\newcommand{\eml}{\end{gather}}

\DeclareMathOperator{\Ran}{Ran}

\newcommand{\ran}{\rangle}
\newcommand{\lan}{\langle}

\newcommand{\beq}{\begin{equation}}
\newcommand{\eeq}{\end{equation}}
\newcommand{\RE}{\mathbb{R}}

\def\CO{{\mathbb C}}

\renewcommand{\leq}{\leqslant}
\renewcommand{\geq}{\geqslant}
\renewcommand{\epsilon}{\varepsilon}

\renewcommand{\d}{{\bf d}}
\newcommand{\x}{{\bf x}}

\renewcommand{\k}{{\bf k}}
\newcommand{\0}{{\bf 0}}

\newcommand{\Rb}{{\bf R}}

\newcommand{\FF}{{\F^\sharp}}

\newcommand{\V}{\mathcal V}

\title[Failure of scattering for the NLSE]{Failure of scattering for the NLSE with a point interaction in dimension two and three}

\author[C. Cacciapuoti]{Claudio Cacciapuoti}
\address{Dipartimento di Scienza e Alta Tecnologia, Universit\`a dell'Insubria, Via Valleggio 11, 22100 Como, Italy}
 \email{claudio.cacciapuoti@uninsubria.it}

\author[D. Finco]{Domenico Finco}
\address{Facolt\`a di Ingegneria, Universit\`a Telematica
Internazionale Uninettuno,  Corso Vittorio Emanuele II 39, 00186 Roma, Italy}
\email{d.finco@uninettunouniversity.net}

\author[D. Noja]{Diego Noja}
\address{Dipartimento di Matematica e Applicazioni, Universit\`a
 di Milano Bicocca,  via Roberto Cozzi 55, 20125 Milano, Italy}
\email{diego.noja@unimib.it} 

\begin{document}

\begin{abstract}

In this paper we consider the NLS equation with power nonlinearity and a point interaction (a ``$\delta$-potential'' in the physical literature) in dimension two and three. We will show that for low power nonlinearities there is failure of scattering to the free dynamics or to standing waves. In the recent paper \cite{MN} Murphy and Nakanishi consider the NLS equation with  potentials and measures, singular enough to include the $\delta$-potential in dimension one and they show analogous properties. We extend the result to higher dimensions and this needs a different treatment of the linear part of the interaction, due the qualitatively different and stronger character of the singularity involved.
\end{abstract}

\maketitle

\begin{footnotesize}
 
 
\end{footnotesize}

\section{Introduction}
We consider the nonlinear Schr\"odinger equation
\begin{equation} \label{cauchy}
\begin{aligned}
i \partial_t\psi   &= \HH_\alpha \psi  + F(\psi) \\
 \end{aligned} 
\end{equation}
where $\HH_\alpha$ is the family of self-adjoint operators known as point interactions (see later) and $F(\psi)=\pm|\psi|^{p-1}\psi$ .  This equation has been recently considered, in dimension two and three, by several authors as regards its well posedness in different functional frameworks (\cite{CFN21, FGI22}), existence and stability of standing waves (\cite{ABCT_jmp, ABCT_CV}) and blow-up (\cite{FN_22}).
When $\alpha = +\infty$ the operator $\HH_\alpha$ coincides with the Laplacian and \eqref{cauchy} is the standard NLS equation. In this case it is well known that scattering to the free dynamics is impossible for low power nonlinearity, in the sense that no non trivial initial state behaves asymptotically as a solution of the free equation (``absence of scattering''). The nonlinearities with this property are called long range nonlinearities, because in a sense they mimic the effect of a long range potential in the linear equation. They are given by the condition $1<p\leq 1+\frac{2}{n}$, corresponding to $1<p\leq 2$ for $n=2$ and $1<p\leq \frac{5}{3}$ for $n=3$. Absence of scattering in this framework has been well known since a long time, starting with the original papers \cite{Gla, Stra, Ba}; see also \cite{caz}, Theorem 7.5.2. Quite recently the subject has been reconsidered by several authors. In particular, a wide generalization of the classical quoted results has been obtained by Murphy and Nakanishi in \cite{MN}. In the cited paper the authors generalize the absence of scattering for long range nonlinearities in two directions. They replace the Laplacian with a Schr\"odinger operator with quite general external potentials, including measures; and when the dynamics allows the existence of space localized solutions globally bounded in time (such as, e.g., standing waves), they also exclude the occurrence of  scattering around these solutions. Among the models included, the authors stress the case of the NLS equation with a delta potential in one dimension, a quite well studied model in the last decades, with many results ranging from well posedness, to scattering and orbital and asymptotic stability of standing waves (see the bibliography in \cite{MN}). Further extensions of the result for the NLS equation with a delta potential have also been given in the (still one dimensional) case of the NLS equation on star graphs \cite{aoki1, aoki2}. In particular it is implicit in \cite{aoki1} the case of a line with a general point interaction, not restricted to delta potentials. Concerning the scattering problem for  NLS equation on graphs, we also point out the recent paper \cite{ikeda22} in which the authors analyze the scattering vs. blow-up dichotomy in the supercritical regime ($p>5$) for the NLS equation on a star-graph. In the present paper we consider the higher dimensional versions of a delta potential, the so called point interactions. As it is well known, in dimension higher  than one there is no possibility to properly define a $(-\Delta+\delta)$-like operator, at least if one wishes to give to the delta distribution the usual meaning. The way out is to define a self-adjoint operator corresponding to a perturbation of the Laplacian concentrated at a point. A natural way to do this is to consider the symmetric operator $-\Delta$ with domain $D(-\Delta) := C_0^\infty(\RE^n\backslash\{\0\})$ and to build  its self-adjoint extensions. A non trivial result of this procedure (i.e. not coinciding with the Laplacian) only exists if $n\leq 3$, and for the case of our interest ($n=2,3$), it constitutes a $1$-parameter family of self-adjoint operators $\HH_\al$ where  $\alpha\in \RE^*$ is the parameter that fixes the self-adjoint extension; as already mentioned, for $\alpha=+\infty$ one has $\HH_\al=-\Delta$ (see \cite{Albeverio} and Section 2 for more details).\\ With this premise, Eq. \eqref{cauchy} is meaningful, as an abstract Schr\"odinger equation with a well defined linear part and as  recalled, several well posedness results exist{, see, e.g., \cite{CFN21, FGI22} and, for abstract results, \cite[Ch. 3.3]{caz} and \cite{OSY12}}. Concerning the behavior of its solutions,
our main result is the following.
\begin{teo} \label{t:main} Suppose that a solution $\psi(t)$ of Eq. \eqref{cauchy} exists in $L^2$, conserves the $L^2$-norm, it is global in the future and that it admits the asymptotic decomposition 
\be\label{decomposition}
\psi(t) = e^{i t \Delta} \psi_+ + l(t) + o(1) \qquad \text{in} \;\; L^2(\RE^n) \ \ \text{as}\ \  t\to +\infty 
\ee
with $\psi_+\in L^2(\RE^n)$ and $l\in L^\infty(\RE^+;  L^2 (\RE^n) \cap  L^{q} (\RE^n))$ for some $1<q<2$.\\
Then, if $n=2$ and $1<p{<} 2$, or $n=3$ and   $1<p<4/3$, one has $\psi_+ \equiv 0$.
\end{teo}
\noindent We add some remarks on the statement and on the proof.\\
The meaning of \eqref{decomposition} is that $\|\psi(t) - e^{i t \Delta} \psi_+ - l(t) \|_2\longrightarrow 0$ as $t\to +\infty$. As already mentioned, well posedness of \eqref{cauchy} was treated in \cite{CFN21, FGI22}, but for the validity of the main result, general existence theorems are not needed once the hypotheses are satisfied.\\ 
We also remark that the main result could be stated with $e^{-i\HH_\alpha t}$ in the place of $e^{i\Delta t}$ thanks to scattering theory for the couple $(\HH_\alpha, -\Delta)$ (see the proof of Theorem \ref{t:main} in Section 3).\\
Notice also that the value or sign of $\alpha$ plays no role in the statement of the result.\\
Put in an informal way, Theorem \ref{t:main} states that  for the NLS equation with a point interaction and a long range power nonlinearity, asymptotically free states ($l\equiv 0$) do not exist, nor do exist states that are asymptotic to some kind of localized structure ($l\neq0$) up to a free evolution. As recalled at the beginning, for this model standing waves exist (at least in the focusing case $F(\psi)=-|\psi|^{p-1}\psi$), so that the presence of the $l(t)$ component in the decomposition is essential. Other kind of localized solutions could also exist (for example in the one dimensional case breathing solutions do  exist) but the analysis of models with point interactions in higher dimension is at its beginnings and nothing else is known.\\
Concerning the proof, its skeleton is the same as in the original Glassey's paper \cite{Gla}, after taking into account the improvement of \cite{MN}. {However, differently from \cite{MN}, point interactions in dimensions two and three cannot be treated perturbatively.  As a main consequence, one has  the limitation to $1<p<4/3$ if $n=3$ (in contrast to $1<p\leq5/3$ as in \cite{MN}).   At a technical level the issue arises from the boundedness of the wave operators in $L^p(\RE^3)$ only for $1<p<3$ for point interactions in dimension three (see Remark \ref{remark5} for details). In the light of these considerations, a possible extension to higher powers of the nonlinearity requires an entirely  different analysis.}
%
%
%
\\
Notice finally that the analogous result holds true when $t\to-\infty$, with the same proof replacing the wave operator $W^+$ with $W^-$. However, to avoid irrelevant complications in the notation, we will denote the wave operator with $W$ without any superscript, we only give the proof of the theorem as it is stated, for $t\to+\infty$, and in the analysis it is understood $W=W^+$.\\The structure of the paper is as follows. In Section 2 firstly definition and properties of point interactions are very briefly recalled; the rest of the section is devoted to discuss some scattering properties of the couple $(\HH_\alpha, -\Delta)$. While the material is simple, to the knowledge of the authors it is not stated elsewhere, so it is here included with some detail.
Section 3 is devoted to the proof of Theorem \ref{t:main}. 
\section{Preliminaries \label{s:Prel}}
\noindent In this section we fix the notation, give the main definitions and prove several technical results that will be used in what follows.
\subsection{Notations}
We denote by $\x$, $\k$ and so on, points in $\RE^n$, $n=2,3$.\\ 
We denote by $\hat f$ or $\F f$ the Fourier transform of $f$, defined to be unitary in $L^2(\RE^n)$: 
\[
\hat f(\k) := \frac1{(2\pi)^{n/2}} \int_{\RE^n} \d\x\, e^{-i\k\cdot\x} f(\x) \qquad \k \in \RE^n. 
\]
We denote by $\|\cdot\|$ the $L^2(\RE^n)$-norm associated with the inner product $\langle\cdot, \cdot\rangle$ and with $\|\cdot\|_p$ the $L^p(\RE^n)$-norm while we use $\| \cdot \|_{H^s}$ for the norm in the Sobolev spaces $H^s( \RE^n)$, $s\in\RE$. The set of bounded operators between two Banach space $X$ and $Y$ is denoted by $\B(X,Y)$.\\
For all $\lambda>0$ we denote by $G^\lambda$ the $L^2$ solution of the distributional equation $(-\Delta +\lambda^2 )G^\lambda = \delta_{\0}$, where $\delta_{\0}$ is the Dirac-delta distribution centered in $\x=\0$.  Explicitly we have: 
\[G^\lambda ( \x ) =\left\{ \begin{aligned}
& \frac{1}{2\pi} K_0 ({\la}\,|\x|) \qquad & n=2; \\ 
& \frac{ e^{-{\la} \, |\x|}}{4\pi |\x| } & n=3,
\end{aligned}\right.
\]
where  $K_0$ is the Macdonald function of order zero. We recall the relation $\frac{1}{2\pi} K_0(z) = \frac{i}{4} H_0^{(1)}(iz )$, where $H_0^{(1)}(z)$ is the Hankel function of first kind and order zero (also known as zero-th Bessel function of the third kind), see, e.g., \cite{was} Eq. (8) p. 78). For real arguments, we will also need the relation $H_0^{(1)} (-|\k||\x| )=- \overline{H}_0^{(1)} (|\k||\x| )$ (\cite {was}).\\
Given $q\in[1,+\infty]$, $q'$ denotes the conjugate exponent, i.e., $q^{-1} + {q'}^{-1} = 1$.\\
For a given real number $q$, we write $q\pm$ to denote $q\pm \ve $  for some sufficiently small $\ve>0$.\\
We denote by $o(L^q)$ any term whose norm in $L^q(\RE^n)$ {converges to $0$} as $t\to \infty$.\\
We use $c$ and $C$ to denote generic positive constants whose  dependence on the parameters of the problem is irrelevant and it is understood that their value may change from line to line.

\subsection{Point interactions\label{ss:PI}} We denote by $\HH_\alpha$  the  self-adjoint operator in $L^2 (\RE^n)$, $n=2,3$, given by the Laplacian with a delta interaction of ``strength'' $\alpha$ placed in $\x =\0$.

We recall that, see \cite{Albeverio}, both for $n=2$ and $n=3$ the structure of the domain of $\HH_\alpha$ is the same: 
\begin{equation}\label{DD}
D (\HH_\alpha) = \left\{ \psi \in L^2 (\RE^n)|\;\psi  = \phi^\la   +q \,  G^\lambda ,\, \phi^\la \in H^2(\RE^n){, \, q\in\CO}, \;\;  q= {\Gamma_\alpha(\la)}  \phi^\la(\0)\ri\}
\end{equation}
with 
\[
\Gamma_\alpha(\la) = \left\{\begin{aligned}
& \f{2\pi}{2\pi \al +\ga +\ln({\la}/2) } \qquad  & n=2, \\
& \f{1}{ \al +\frac{{\la}}{4\pi}  } & n=3 ;
\end{aligned} \right. \qquad \alpha \in\RE. 
\]
Here  $ \la $ can be taken in $\RE^+$, possibly excluded one point which we denote by $|E_\alpha|^{1/2}$, where $E_\alpha$ is the negative eigenvalue of $\HH_\alpha$, see below for the details. {Even though $\lambda$ enters the right hand side in the definition of the operator domain $D (\HH_\alpha)$, since $G^\lambda - G^{\tilde \lambda} \in H^2(\RE^n)$ for any $\lambda,\tilde\lambda>0$, the domain does not depend on $\lambda$, see, e.g., \cite[Pg. 261]{FGI22}. Moreover, $q$ is determined only by the behavior of $\psi$ as $|\x|\to0$ and does not depend on $\lambda$. Explicitly, 
\[
q = - \lim_{|\x|\to0}\frac{2\pi}{\ln(|\x|)} \psi(\x) \qquad \text{for $n=2$}\] and
\[
q = \lim_{|\x|\to0}4\pi |\x| \psi(\x) \qquad \text{for $n=3$}.\]
}
For $n=2$,  $\gamma$ is the Euler-Mascheroni constant (for a comparison with \cite{Albeverio}, we recall that $\gamma$ is related to the digamma function $\Psi$ by $\gamma = - \Psi(1)$). The constant $\al$ is real and it parametrizes the family of operators through the relation  $ q= \Gamma_\alpha(\la)  \phi^\la(\0)$, which plays the role of a boundary condition at the singularity. For both $n=2$ and $n = 3$ the free dynamics is recovered in the limit $\alpha\to+\infty$. {Taking into account the decomposition $\psi  = \phi^\la   +q \,  G^\lambda $ in the definition of the  operator domain, t}he action of the operator  is given by 
\begin{equation}\label{action0}
\HH_\alpha \psi = -\Delta \phi^\la - \lambda^2 q G^\lambda \qquad \forall\psi\in D(\HH_\alpha).
\end{equation}

\n
The spectrum of the Hamiltonian $\sigma_c(\HH_\alpha)=\sigma_{ac}(\HH_\al)=[0,\infty)$ and $\sigma_{sc}(\HH_\al)=\emptyset$; for $n=2$, $\HH_\alpha$ has  a simple negative eigenvalue $\{ E_\al \}$ for any $ \al \in \RE$; for $n=3$,  if $\al\geq 0$ there is no point spectrum, while for $\al<0$ there is a simple negative eigenvalue $\{ E_\al \}$. Whenever the eigenvalue $E_\alpha$ exists, we denote by $\Phi_\alpha$ the corresponding normalized eigenvector. We will not need the explicit form of eigenvalues (see \cite{Albeverio} for details), while eigenvectors are given by
\[
\Phi_\al (\x) =  \left\{ \begin{aligned}
&     N_\alpha  K_0 (  2\, e^{-(2\pi \al +\ga)}|\x|) \qquad && n=2,\; \alpha\in\RE; \\ 
&  \sqrt{2|\alpha|}\frac{ e^{4 \pi \al \, |\x|}}{|\x|}  &&   n=3,\;\alpha <0,
\end{aligned}\right.
\]
where $N_\alpha$ is a normalization constant whose explicit value is irrelevant for our analysis. 

To simplify the notation, and since $\alpha$ is regarded as a fixed parameter, we omit the suffix  	$\alpha$ from objects that may depend on it  and from now on we simply write, for example,  $\HH\equiv \HH_\alpha$. 

\subsection{Generalized Fourier Transform\label{ss:2.2}}
In this section we prove several results regarding the generalized Fourier Transform associated with the continuous spectrum of $\HH$.

We start with with the well known Dollard decomposition of the free Schr\"odinger dynamics
\be \label{decomp}
e^{i t\Delta } = M_tD_t \F M_t,
\ee
where the unitary operators $M_t$ and $D_t$ are given by
\be\label{MD}
M_t f(\x) = e^{i \f{{|\x|}^2}{4t} }f(\x), \qquad D_tf (\x) = \frac{1}{(i 2t)^{n/2}  } f \left( \frac{\x}{2t} \right).
\ee
Formula \eqref{decomp} in particular implies that
\beq \label{larget}
{\lim_{t\to +\infty} \|(e^{i t\Delta }  - M_t D_t\F) f\| =0 \qquad \forall f\in L^2(\RE^n) },
\eeq
by the unitarity of $M_tD_t \F$ and dominated convergence.
We extend formula \eqref{larget} to more general situations. We present the following result which extends formula \eqref{larget} to an abstract setting and may have an interest on its own. {We recall that for an Hilbert space $\mathscr{H}$,  $V\in \B(\mathscr H)$ is a partial isometry if it is an isometry on the orthogonal complement of its kernel.}

\begin{prop} \label{fintagen}
Let $H,\, H_0$ be self-adjoint operators in an Hilbert space ${\mathscr H}$. Assume that there exist a unitary operator $ \F$ and a
partial isometry $ \F^\sharp$ such that
\be \label{wavefinti}
\lim_{t\to +\infty} \| e^{i {t} H } e^{-i t H_0} f - {\FF}^\ast \F f \| =0 \qquad \forall f \in \mathscr H.
\ee
If there exist a unitary operator $\V_t$, such that
\be \label{bah}
\lim_{t\to +\infty} \|  e^{-i t H_0} f - \V_t \F f \| =0 \qquad \forall f \in \mathscr H, 
\ee
then
\be \label{fintagenerale}
\lim_{t\to +\infty} \|  e^{-i t H} f - \V_t \FF f \| =0 \qquad \forall f \in \Ran ( \FF^\ast). 
\ee
\end{prop}
\begin{proof}
Let $g \in \mathscr H$, then Eq. \eqref{wavefinti} is equivalent to 
\be
\lim_{t\to +\infty} \| e^{i t H } e^{-i t H_0} \F^\ast g - {\FF}^\ast g \| =0  ,
\ee
and 
\be
\lim_{t\to +\infty} \|  e^{-i t H_0} \F^\ast g - e^{-i tH }{\FF}^\ast g \| =0 .
\ee
By \eqref{bah} we have
\be
\lim_{t\to +\infty} \| \V_t g - e^{-i t H }{\FF}^\ast g \| =0 .
\ee
Setting $f = {\FF}^\ast g$ we arrive at \eqref{fintagenerale}.
\end{proof}

In concrete cases one has ${\mathscr H}=L^2(\RE^n)$, $H_0 = -\Delta$, $H=-\Delta + V$, and $\F$ is the Fourier transform.
Taking into account \eqref{larget} we can set $\V_t =M_t D_t$.
The role of $\FF$ is played by the generalized Fourier transform. Let $\Phi (\k, \x)$ be the generalized eigenfunctions of $H$ (see e.g. \cite {RSIII} Section XI.6), and let
\be \label{genfou}
\FF f (\k) = \int_{\RE^n} \overline{ \Phi (\k, \x) } f(\x) \d\x \qquad \k\in\RE^{n}.
\ee
Assuming enough regularity for $V$, see \cite{Ag, Ike, RSIII} for different sets  of hypothesis,
it is known that $\Phi (\k, \x)$ are well defined, Eq.\eqref{genfou} defines a partial isometry on $L^2(\RE^n)$,  and that 
$\text{Ran} (\FF^\ast) = P_{ac}(H) L^2(\RE^n)$. Moreover the wave operators $W$ exist, they are complete and
\[
W\, f  = \lim_{t\to +\infty} e^{i {t} H } e^{-i t H_0} f = {\FF}^\ast \F f .
\]
In such cases all the hypothesis of Proposition \ref{fintagen} are satisfied and  for any $f\in  P_{ac}(H) L^2(\RE^n)$
one has 
 \[
\lim_{t\to +\infty} \|  e^{-it H} f - M_t D_t \FF f \| =0.
\]
In what follows we will often use the notation $\FF f = f^\sharp$.

Proposition \ref{fintagen} applies  also when  $H= \HH $, $n=2,3$, a situation not immediately included in \cite{Ag, Ike, RSIII}, thanks to the fact that
in these cases the generalized eigenfunctions  and the wave operators have an explicit expression. The generalized eigenfunctions are given by 
\begin{equation}\label{geneigen}
\Phi (\k, \x) =
\frac{e^{i \k\cdot \x}}{(2\pi)^{n/2}}  + \Rb(\k,\x)  
\end{equation}
with  
\begin{equation}\label{geneigenR}
\Rb(\k,\x)=\left\{	\begin{aligned}	
		&\frac{i}4 \frac{1}{2\pi \alpha +\gamma +\ln ({-}|\k|/2i)} H_0^{(1)} (-|\k||\x| ) \qquad &&  \text{for } n=2\\ 
		&\frac{1}{(2\pi)^{3/2}} \frac{1}{4\pi\alpha +i|\k|} \frac{e^{-i|\k||\x|}}{|\x|} && \text{for } n=3,
		\end{aligned} \right.	
\end{equation}
where the logarithm has to be understood as its principal value: $\ln (|\k|/2i)=\ln (|\k|/2)-i{\pi}/{2}${, see \cite[Ch. I.1.4 - Eq. (1.4.11) and Ch. I.5 - Eq. (5.37)]{Albeverio}}. The generalized Fourier transform defined via \eqref{genfou} is a partial isometry on $L^2(\RE^n)$ such that $\text{Ran} (\FF^\ast) = P_{ac}(\HH) L^2(\RE^n)$. {From \cite[App. E]{Albeverio} there follows that} the wave operators $W$ exist, they are complete and can be written as:
\[
W = I + \Omega,
\]
where $\Omega$ is the operator (mapping to radial functions)
\[
\Omega f (\x) =\Omega f (|\x|)=\begin{dcases}
 \dfrac{1}{i\pi} \int_0^\infty \GG_{-\la}(\x) \dfrac{2\pi \la}{2 \pi\al + \ln\lf( \frac{\la}{2} \ri) {+}\frac{\pi i}{2} + \ga } \lf( \int_{\mathbb S^1} \F f(\la \ome ) d\ome \ri) d\la & n=2 \\
\frac{1}{i(2\pi)^{2} } \int_0^\infty \GG_{-\la}(\x) \dfrac{ \la}{\al +i\frac{\la}{4\pi} } \int_\RE e^{i\la r } r \lf( \int_{\mathbb S^2} f(r \ome ) d\ome \ri)dr\, d\la & n=3,
\end{dcases}
\]
and we denoted
\be \label{zoccola}
\GG_{\la}(\x) = 
\begin{cases} 
H_0^{(1)} (\la |\x|) & n=2 \\
\dfrac{e^{i\la |\x|}}{|\x|} & n=3.
\end{cases}
\ee
Finally, we have $W=\FF^\ast \F$. {See also \cite[Eq. (98)]{Y2} and \cite[Eq. (3.6)]{Y1}.} Therefore, thanks to Proposition \ref{fintagen} for any $f\in  P_{ac}(\HH) L^2(\RE^n)$ there holds
\be \label{utile}
\lim_{t\to +\infty} \|  e^{-i t \HH} f - M_t D_t \FF f \| =0  ,
\ee

\noindent
where $\FF$ denotes the generalized Fourier transform defined by Eqs. \eqref{genfou} and \eqref{geneigen}.\\ We conclude this section of preliminaries with the following proposition, showing that $\FF$ satisfies a Hausdorff-Young type inequality.
\begin{prop}Let $\FF$ be defined by Eqs. \eqref{genfou} and \eqref{geneigen}. Then, for any $q \in[2,\infty)$ if $n=2$ and for any $q\in[2,3)$ if $n=3$ there holds true
\be \label{genhy}
\| \FF f \|_q \leq c \| f\|_{q'}. 
\ee
\end{prop}
\begin{proof}
We start from $W= \FF^\ast \F$ where $W$ is the wave operator. Hence,  $\FF^\ast = W \F^\ast$ and $\FF  = \F W^\ast$.
It is known, see \cite{Y1,Y2}, that $W: L^2 (\RE^n) \to  L^2 (\RE^n)$ extends to a bounded operator $W: L^r (\RE^n) \to  L^r (\RE^n)$ 
for $r \in (1,\infty)$ if $n=2$ and for $r \in (1,3)$ if $n=3$. Hence ,  $W^*: L^{r'} (\RE^n) \to  L^{r'} (\RE^n)$  for $r' \in (1,\infty)$ if $n=2$ and for $r' \in (3/2,\infty)$ if $n=3$. Then \eqref{genhy} follows from the Hausdorff-Young inequality
\[
\| \FF f \|_q = \|  \F W^\ast  f \|_q   \leq  \| W^\ast  f \|_{q'} \qquad q\geq 2.  
\]
\end{proof}

 \section{Proof of Theorem \ref{t:main}}

First we prove  a technical lemma.
 \begin{lem}\label{l:stima-lineare}  
 Let  $\varphi_0 \in C^\infty_0 (\RE^n \setminus \{ \0 \} )$, set $\varphi:= \varphi_0 - \Phi_\alpha \langle \Phi_\alpha,\varphi_0\rangle$ ($\varphi:= \varphi_0$ if $\HH$ has no point spectrum) and define $\tilde w(t ) := D_t^\ast M_t^\ast e^{-i t \HH } \varphi$.  Then: \\
 $i)$ for  $q\geq 2$ if $n=2$;\\
 $ii)$ for $2\leq q<3$ if $n=3$;\\
 there holds 
 \begin{equation}\label{stima-lineare}
 \tilde w(t)=  \varphi^\sharp + o(L^q ) \qquad \text{as }t\to+\infty , 
\end{equation}
{where we recall that $\varphi^\sharp = \FF\varphi$}. 
 \end{lem}
 \begin{proof}
{We start by noticing that $\varphi = P_{ac}(\HH)\varphi $ by definition.} Using the intertwining property of the wave operator{, i.e. $ e^{-i\HH t} P_{ac}(\HH) =W  e^{i\Delta t} W^\ast$,} and \eqref{decomp}, we have
\begin{align}
\tilde w (t) & =D_t^* M_t^* e^{-i\HH t} P_{ac}(\HH)\varphi      \nonumber \\
& =D_t^* M_t^*W  e^{i\Delta t} W^\ast \varphi      \nonumber \\
& =D_t^* M_t^* W  M_t D_t \F M_t W^\ast \varphi      \nonumber \\
& =D_t^* M_t^* (I+\Omega)  M_t D_t \F M_t W^\ast \varphi      \nonumber \\
& = \F M_t W^\ast \varphi  + D_t^* M_t^* \Omega  M_t D_t \F M_t W^\ast \varphi.      \label{albero} 
\end{align} 
Now let us first consider the two dimensional case.\\ 
Notice that for $ 2\leq q <\infty$ we have
\[
\lim_{t\to +\infty} \| \F (M_t-1) W^\ast \varphi \|_{q}= \lim_{t\to \infty} \|  (M_t-1) W^\ast \varphi \|_{q'}=0
\]
thanks to $ \varphi \in L^{q'}(\RE^2)$, $ W^\ast \in \B(L^r(\RE^2), L^r(\RE^2) )$ for $r\in( 1,\infty)$ (see \cite{Y2}) and dominated convergence theorem.
Therefore, keeping into account that $\F W^\ast= \F^\sharp$, we can write:
\begin{align}
\F M_t W^\ast \varphi  & = \F W^\ast \varphi + \F (M_t-1) W^\ast \varphi \nonumber \\
&= \varphi^\sharp  +o(L^q). \label{decorazioni}
\end{align}
For $q=2$, from \eqref{albero} and \eqref{decorazioni} and taking into account that from \eqref{utile} it follows $\tilde w(t) = \varphi^\sharp + o(L^2)\ ,$ we conclude that
\be \label{stella}
D_t^* M_t^* \Omega  M_t D_t \F M_t W^\ast \varphi  = o(L^2).
\ee
{Next we prove that $\|D_t^* M_t^* \Omega  M_t D_t \F M_t W^\ast \varphi \|_r \leq c \|\varphi\|_{r'}$, for any $r$ such that $q<r<\infty$, where $c$ does not depend on $t$. To this aim} we need some detailed properties of the operators $\Omega$.
In facts, {see \cite{Y2},} $\Omega$ can be written as a composition of two operators: $\Omega f = K m (|D|) f$ where 
\[
K f(|\x|) = \frac{1}{i\pi} \int_0^\infty \GG_{-\la}(\x) \la
\lf( \int_{\mathbb S^1} \F f(\la \ome ) d\ome \ri) d\la
\]
and
\[
 m (|D|) f (\x)  = \frac{1}{2\pi } \int_{\RE^2} e^{i \x\cdot \k} m (|\k|) \F f(\k)\, \d\k\ ,
\qquad 
m (|\k|)=
\frac{2\pi }{2 \pi\al + \ln\lf( \frac{|\k|}{2} \ri) {+}\frac{\pi i}{2} + \ga }.
\]
Both the operators belongs to $\B(L^r(\RE^2), L^r(\RE^2) )$ for $r\in( 1,\infty)$. We derive now a commutation property of $\Omega $.
Using the explicit form of generalized eigenfunctions \eqref{zoccola}, we have
\begin{align}
D^\ast_t \Omega f(|\x|) 
& = \frac{t}{i\pi} \int_0^\infty \GG_{-\la}(t\x) \frac{2\pi \la}{2 \pi\al + \ln\lf( \frac{\la}{2} \ri) {+}\frac{\pi i}{2} + \ga }
\lf( \int_{\mathbb S^1} \F f(\la \ome ) d\ome \ri) d\la  \nonumber \\
& = \frac{1}{ it\pi} \int_0^\infty \GG_{-\la}( \x) \frac{2\pi \la}{2 \pi\al + \ln\lf( \frac{\la}{2t} \ri) {+}\frac{\pi i}{2} + \ga }
\lf( \int_{\mathbb S^1} \F f \lf(\frac{\la \ome}{t} \ri) d\ome \ri) d\la  \nonumber \\
& =  \frac{1}{ i\pi} \int_0^\infty \GG_{-\la}( \x) \frac{2 \pi \la}{2 \pi\al + \ln\lf( \frac{\la}{2t} \ri) {+}\frac{\pi i}{2} + \ga }
\lf( \int_{\mathbb S^1} \F D^\ast_t f(\la \ome ) d\ome \ri) d\la  \nonumber \\
&= K  m_t (|D|) D^\ast_t f (|\x|)
\end{align}
where we introduced the operator 
\[
 m_t (|D|) f (\x)  = \frac{1}{2\pi } \int_{\RE^2} e^{i \x\cdot \k} m_t (|\k|) \F f(\k)\,\d\k\ ,
\qquad 
m_t (|\k|)=
\frac{2\pi }{2 \pi\al + \ln\lf( \frac{|\k|}{2} \ri)-\ln(t) {+}\frac{\pi i}{2} + \ga}.
\]
Notice that $m_t\to 0$ pointwise but not uniformly, so
$
\lim_{t\to \infty}  \| m_t (|D|)  \|_{\B(L^q (\RE^2), L^q (\RE^2) )} =0
$
fails. However, exploiting Mikhlin's multiplier theorem, the operator $ m_t (|D|) $ is bounded in $L^q(\RE^2)  \ \forall q\in (1,\infty)$:
\be \label{pappa}
 \| m_t (|D|)  \|_{\B(L^q (\RE^2), L^q (\RE^2) )} \leq c\ ,
\ee where the constant $c$ is uniform in $t$ (see \cite{Y2} for a similar analysis and more details, in particular Lemma 4.3 and Section 4.4). Fix now $q<r<\infty$. One has
\begin{align}
\|  D_t^* M_t^* \Omega  M_t D_t \F M_t W^\ast \varphi  \|_r
& = \|  K m_t(|D|)  D_t^*  M_t D_t \F M_t W^\ast \varphi  \|_r \nonumber \\
& = \|  K m_t(|D|)    M_{1/t}  \F M_t W^\ast \varphi  \|_r \nonumber \\
& \leq c \|    \F M_t W^\ast \varphi  \|_r \nonumber \\
& \leq c \|    W^\ast \varphi  \|_{r'} \nonumber \\
& \leq c \|    \varphi  \|_{r'} \label{palle}.
\end{align}
We  conclude that $ D_t^* M_t^* \Omega  M_t D_t \F M_t W^\ast \varphi  =o(L^q)$ by \eqref{stella}, \eqref{palle} and interpolation.  Giving this and taking into account \eqref{albero} and \eqref{decorazioni} the thesis
 \eqref{stima-lineare} finally follows for the two dimensional case.

\noindent
The argument can be repeated in the three dimensional case with minor modifications. We claim that 
$\Omega f = K m (|D|) f$ where 
\[
K f(|\x|) = \frac{i}{(2\pi)^{2} } \int_0^\infty \GG_{-\la}(\x) \int_\RE e^{-i\la r} r Af(r) dr\, d\la\ ,
\qquad
Af(r)=  \int_{\mathbb S^2}  f(r {\ome} ) d{\ome} 
\]
and
\[
 m (|D|) f (\x)  = \frac{1}{(2\pi)^{3/2} } \int_{\RE^3} e^{i \x\cdot \k} m (|\k|) \F f(\k)\, \d\k\ ,
\qquad 
m (|\k|)=
\frac{|\k|}{\al + \frac{i |\k|}{4\pi}  }.
\]
To check that this is indeed the case we start by pointing out  the following  identities which  can be proved by a straightforward calculation:
\[
A\F f(r) = - \frac{1}{i r (2\pi)^{1/2}} \int_\RE  e^{-irs} sAf(s) ds\;;\qquad 
A\F^{-1} f(r) = \frac{1}{i r(2\pi)^{1/2}} \int_\RE  e^{irs} sAf(s) ds.
\]
 We also point out the trivial identity  $A\F f(r) = A\F^{-1} f(r)$, even though we will not use it. Next we compute 
\[
rAm (|D|) f (r) = 
 \frac{1}{i(2\pi)^{1/2}} \int_\RE  e^{irs} s m(s) A \F f(s) ds = 
\frac{1}{2\pi} \int_\RE  e^{irs}  m(s) 
  \int_\RE  e^{-iss'} s'Af(s') ds' ds,
\]
where in the first identity we used the fact that the function  $m$ is spherically symmetric, hence   $Am(\cdot) \F f(s) = m(s) A \F f(s)$. The decomposition $\Omega f = K m (|D|) f$ follows from the inversion formula for the one dimensional Fourier transform. 

The operator $m(|D|) $ is bounded in $L^r(\RE^3)$ for $r\in(1,\infty)$ again by Mikhlin's multiplier theorem.  Following and adapting \cite{Y1}, the operator $K$ belongs to $\B(L^r(\RE^3), L^r(\RE^3) )$ for $r\in( 1,3)$ and consequently $\Omega$ belongs to $\B(L^r(\RE^3), L^r(\RE^3) )$ for $r\in( 1,3)$.  Here we give a brief proof of $K\in B(L^r(\RE^3), L^r(\RE^3) )$ when $r\in( 1,3)$, for the reader's sake and because of slight differences from the treatment in \cite{Y1}.
By the change of variable $r \to -r$, the operator $K$ can be written as 
\[
K f(|\x|) =- \frac{i}{(2\pi)^{2} } \frac{1}{|\x|} \int_0^\infty e^{-i\lambda |\x|} \int_\RE e^{i\la r} r Af(r) dr\, d\la.
\]
Hence, by the properties of the (one dimensional) Fourier transform and after straightforward computations, it can be rearranged as 
\[
K f(|\x|) = -\frac{1}{(2\pi)^{2} |\x| } \lf( L\ast r Af \ri) (|\x|)\ , \qquad L(x) = \frac{1}{x-i0},
\]
where we introduced the distributional Calderon-Zygmund kernel $L$ (see \cite{SS11}, Section 3.1).
It is well known that
\[
\int_\RE | (L\ast g)(x)|^p \, \rho(x) dx \leq \int_\RE |  g(x)|^p \, \rho(x) dx 
\]
when $\rho$ belongs to the Muckenhoupt   class of $A_p$-weights (see \cite{G1}, Theorem 7.4.6). In our case $r^a \in A_p$ for $-1<a<p-1$ (see \cite{G1}, Example 7.1.7).
For $p \in (3/2,3)$ we have
\[
\| Kf\|_p^p = c \int_0^\infty r^{2-p} | L\ast r Af  (r) |^p dr \leq  c \int_0^\infty r^{2-p} |  r Af  (r) |^p dr = c \int_0^\infty r^{2} |   Af  (r) |^p dr =c \| f\|_p^p.
\]
Integrating by parts, we also have
\[
K f(|\x|) =- \frac{1}{(2\pi)^{2\\
	\\
} |\x|^2 } \lf( L\ast r^2 Af \ri) (|\x|), 
\]
and repeating the above argument it follows that for $p\in (1,3/2)$ we have
\[
\| Kf\|_p^p = c \int_0^\infty r^{2-2p} | L\ast r^2 Af  (r) |^p dr \leq  c \int_0^\infty r^{2-2p} |  r^2 Af  (r) |^p dr = c \int_0^\infty r^{2} |   Af  (r) |^p dr =c \| f\|_p^p.
\]
The case $p=3/2$ follows by interpolation. When commuting with dilations, we obtain
\[
D^\ast_t K m(|D|) = K m_t(|D|) D^\ast_t\ , \]
where
\[
 m_t (|D|) f (\x)  = \frac{1}{(2\pi)^{3/2} } \int_{\RE^3} e^{i \x\cdot \k} m_t (|\k|) \F f(\k)\, \d\k\ ,
\qquad 
m_t (|\k|)=
\frac{|\k|}{\al t + \frac{i |\k|}{4\pi}  }.
\]

Since $m_t (|D|)$ is bounded uniformly in $t$, we can repeat the previous argument choosing $2<q<r<3$. The proof of the Lemma is complete.
\end{proof}
\begin{remark} As it is well known, the point interaction in dimension three admits a zero energy resonance for $\al=0$. We notice that this fact does not affect the previous proof of Lemma \ref{l:stima-lineare}.
\end{remark}
\begin{proof}[{\bf Proof of Theorem \ref{t:main}}]
First notice that {(see, e.g., \cite[App. E]{Albeverio})} by the linear scattering theory for $\HH$, for any $\psi_+ \in L^2 (\RE^n)$ there exists $v_+ \in  P_{ac}(\HH) L^2(\RE^n)$ such that 
\[
\lim_{t\to +\infty} \| e^{i t \Delta} \psi_+ - e^{-i t \HH} v_+ \|= 0 .
\]
Therefore it is sufficient to prove that  $\nexists v_+ \in P_{ac}(\HH) L^2(\RE^n)$, $v_+\neq 0$ such that
\be \label{asdecomp}
\psi (t) = e^{-it \HH} v_+ + l(t) + o(L^2) 
\ee
as $t\to +\infty$ with $l\in L^\infty(\RE^+;  L^2 (\RE^n) \cap  L^{2-} (\RE^n))$.
We follow a variant of the classical argument of Glassey \cite{Gla}, due to \cite{MN}; see also \cite{Ba, Stra}.
We proceed by {contradiction}  and assume that  $\exists v_+  \in P_{ac}(\HH) L^2(\RE^n)$ such that \eqref{asdecomp} holds. Let us consider
\be
B(t) = \Im \lan \psi(t) , w(t) \ran
\ee
with  $w(t) = e^{-i t \HH } \varphi$, where  $\varphi $ is a function that satisfies the assumptions of Lemma \ref{l:stima-lineare} and that will be fixed later on. Using \eqref{cauchy} we have
\begin{align}
B(t) & = B(1) + \int_1^t \frac{d}{ds} \Im \lan \psi(s) , w(s) \ran \, ds \nonumber \\
& = B(1) + \int_1^t \Re \lan F(\psi(s)) , w(s) \ran \, ds.
\end{align}
Let us define $\tilde \psi(s ) = D_s^\ast M_s^\ast \psi(s)$ and  $\tilde w(s ) = D_s^\ast M_s^\ast w(s)$. Since $F$ is a power-type nonlinearity, we immediately obtain 
\be
 \lan F(\psi(s)) , w(s) \ran = \frac{1}{s^{\frac{n}{2}(p-1)  } }  \lan F(\tilde \psi(s)) , \tilde w(s) \ran.
\ee
Notice that $B(t)$ is a bounded function since  we 
have $| B(t)  |\leq \| \psi(t) \| \|w(t)\| = \| \psi_0 \| \| \varphi \|$ by Cauchy-Schwartz inequality and conservation of mass.  
Since $1/ s^{\frac{n}{2}(p-1)  }$  is not integrable
at infinity for $n=2,\ p\in (1, 2]$ and for $n=3,\ p\in (1, 5/3]$, if we prove that  $\lim_{s\to \infty}\Re \lan F(\tilde \psi(s)) , \tilde w(s) \ran >0$ we obtain a contradiction and
the proof is complete. However, as we will see, a further condition will be needed for $n=3$.
Indeed, we shall prove that 
\be \label{limit}
\lim_{s\to +\infty} \lan F(\tilde \psi(s)) , \tilde w(s) \ran =  \lan F( v_+^\sharp) , \varphi^\sharp \ran,
\ee 
and then we will show that one can choose $\varphi$ such that  $\Re \lan F( v_+^\sharp) , \varphi^\sharp \ran>0$. In order to prove \eqref{limit} it is sufficient to prove that
\begin{align}
F(\tilde \psi(s)) & = F( v_+^\sharp) + o(L^{\frac{2}{p} } + L^{\frac{2}{p}- }  ) \label{stimauno} \\
\tilde w(s) &=  \varphi^\sharp + o(L^{\frac{2}{2-p}} \cap L^{\frac{2}{2-p}+} ). \label{stimadue}
\end{align}
The asymptotics \eqref{stimadue} follows from Lemma \ref{l:stima-lineare}{, since the constraint  $\frac{2}{2-p} <3$, needed for $n=3$, is satisfied  for $p\in(1,4/3)$}. Next we prove \eqref{stimauno}. Define  $\tilde l(s ) := D_s^\ast M_s^\ast l(s)$. Since we assumed that $ \psi (t)  - e^{-it \HH} v_+ - l(t) = o(L^2)$, 
then $ \tilde \psi (t) - D_t^\ast M_t^\ast e^{-it \HH} v_+ - \tilde l(t) = o(L^2)$. Taking into account \eqref{utile}, we obtain
$ \tilde \psi (t) - v_+^\sharp - \tilde l(t) = o(L^2)$. Moreover, one has
\[
\| \tilde l(s)\|_q^q =
\| D_s^\ast M_s^\ast l(s) \|_q^q = \frac{1}{s^{ \frac{n}{2}(2-q)} } \| l(s)\|_q^q,
\]
and therefore $\tilde l(s) = o(L^{2-} )$. Using H\"older estimate, the conservation of mass and the properties of $\FF$, we obtain 
\begin{align}
\|  F(\tilde \psi(s)) -  F( v_+^\sharp) \|_{L^{2/p} + L^{2/p-} } & \leq c 
\| (|\tilde \psi (s) |^{p-1} + | v_+^\sharp|^{p-1} )( |\tilde \psi (s) - v_+^\sharp - \tilde l(s)| + |   \tilde l(s)|  )  \|_{L^{2/p} + L^{2/p-} } \nonumber \\
& \leq c  ( \|\tilde \psi (s) \|^{p-1} + \| v_+^\sharp\|^{p-1} )( \|\tilde \psi (s) - v_+^\sharp - \tilde l(s)\| + \|   \tilde l(s)\|_{2-} )   \nonumber \\
&\leq c( \|\tilde \psi (s) - v_+^\sharp - \tilde l(s)\| + \|   \tilde l(s)\|_{2-} )  \label{quasiuno},
\end{align}
from which follows the asymptotics \eqref{stimauno}.

The last point is to choose $\varphi$ such that  $\Re \lan F( v_+^\sharp) , \varphi^\sharp \ran >0$.  Since $v_+^\sharp\neq 0 $, then $F(v_+^\sharp) \neq 0$ and one can take  $f\in {C^\infty_0 }(\RE^n\setminus\{0\})$ such that
 $\Re \lan F( v_+^\sharp) , f \ran >0$. Set $g := \FF^\ast f$, by construction $g \in P_{ac}(\HH)$ and   $\langle\Phi_\alpha ,g\rangle =0$ (whenever $\HH$ has non empty point spectrum).  Choose $\varphi_{0} \in C^\infty_0 (\RE^n\setminus \{ \0\} )$ such that 
$\| \varphi_{0} - g \|_{\frac2p}<\ve $ and set $\varphi:= \varphi_{0} - \Phi_\alpha \langle \Phi_\alpha,\varphi_{0}\rangle$, so that $\varphi$ satisfies the assumptions of Lemma \ref{l:stima-lineare}. By noticing that 
\[
| \langle \Phi_\alpha,\varphi_{0}\rangle|  =  |\langle \Phi_\alpha,(\varphi_{0}-g)\rangle| \leq \|\Phi_\alpha\|_{\frac{2}{2-p}}\| \varphi_{0} - g \|_{\frac2p} < \ve \|\Phi_\alpha\|_{\frac{2}{2-p}},
\]
one infers
\[
\| \varphi - g \|_{\frac2p} \leq \| \varphi_{0} - g \|_{\frac2p} + | \langle \Phi_\alpha,\varphi_{0}\rangle|  \|\Phi_\alpha\|_{\frac{2}{p}} \leq \ve (1+ \|\Phi_\alpha\|_{\frac{2}{2-p}}\|\Phi_\alpha\|_{\frac{2}{p}}). 
\]
Then $\| \varphi^\sharp - f \|_{ \frac{2}{2-p} } \leq c \, \ve$ by \eqref{genhy}  and 
$\Re \lan F( v_+^\sharp) , \varphi^\sharp \ran >0$ if $\ve $ is small enough since $ F( v_+^\sharp) \in L^{\frac2p}(\RE^n)$. 
\end{proof}

{\begin{remark}\label{remark5}
Notice that the condition $\frac{2}{2-p}\in [2,3)$ needed to apply \eqref{genhy} forces $p\in [1,4/3)$ for $n=3$, while no further condition beside $p\in(1,2)$ is needed for $n=2$.   Indeed,  for $n=3$ the boundedness of the wave operators, used crucially in the proofs of Eqs. \eqref{genhy} and \eqref{stima-lineare}, is known to fail for $q\geq 3$ (see \cite{Y1}). For a comparison, the simple estimate $(18)$ in \cite{MN} is not available for the full range of exponents; this is due to the fact that the Hausdorff-Young inequality must be replaced by \eqref{genhy}. We remark also that, due to the singularity of $G^\lambda$, elements of $D(\HH)$ in general do not belong to  $L^q(\RE^3)$ for $q\geq 3$. Since the unitary group $e^{-it\HH}$ leaves invariant $D(\HH)$, $\tilde w(t) \notin L^q(\RE^3)$ for $q\geq 3$ and  \eqref{stima-lineare} cannot hold true in this case.  
\end{remark} 
}

\section*{Acknowledgments.}
\noindent The authors acknowledge the support of the Gruppo Nazionale di Fisica Matematica (GNFM-INdAM). D. Noja acknowledges for funding the EC grant IPaDEGAN (MSCA-RISE-778010).

\end{document}